\documentclass[twocolumn,pra,superscriptaddress]{revtex4-2}

\usepackage[utf8]{inputenc}
\usepackage{float}
\usepackage{placeins}
\usepackage{qcircuit}
\usepackage{braket}
\usepackage{amsmath, amsthm}
\usepackage[normalem]{ulem}
\usepackage{todonotes}
\usepackage{xcolor}
\usepackage{multirow}
\usepackage{braket}
\usepackage[normalem]{ulem}
\usepackage{bm}
\usepackage{bbold}

\usepackage{wasysym}
\usepackage{todonotes}

\newcommand{\ketbra}[2]{\mbox{$|#1\rangle\langle #2|$}}

\def\ketbra#1#2{{\vert#1\rangle\!\langle#2\vert}}

\def\braket#1#2{{\langle#1\vert#2\rangle}}

\def\2LH{{\sc $2$-local Hamiltonian}}
\def\5LH{{\sc $5$-local Hamiltonian}}

\usepackage{dsfont}

\newtheorem{theorem}{Theorem}

\newtheorem{remark}{Remark}
\newtheorem{definition}{Definition}

\usepackage{mathtools} % Bonus
\graphicspath{{Figures/}}

\begin{document}

\title{Depth scaling of unstructured search
via quantum approximate optimization}

\author{Ernesto Campos}
\email[e-mail: ]{ernesto.campos@skoltech.ru}
\affiliation{Skolkovo Institute of Science and
Technology, Moscow, Russian Federation}

\author{Daniil Rabinovich}\affiliation{Skolkovo Institute of Science and
Technology, Moscow, Russian Federation}
\affiliation{Moscow Institute of Physics and Technology, Dolgoprudny, Russian Federation}

\author{Alexey Uvarov}
\email{Former affiliation}
\affiliation{Skolkovo Institute of Science and
Technology, Moscow, Russian Federation}

%\homepage[]{quantum.skoltech.ru}
%\altaffiliation{}

\begin{abstract}
Variational quantum algorithms have become the de facto model for current quantum computations. A prominent example of such algorithms---the quantum approximate optimization algorithm (QAOA)---was originally designed for combinatorial optimization tasks, but has been shown to be successful for a variety of other problems. However, for most of these problems the optimal circuit depth remains unknown. One such problem is unstructured search which consists on finding a particular bit string, or equivalently, preparing a state of high overlap with a target state. To bound the optimal QAOA depth for such problem we build 
on its known solution in a continuous time quantum walk (CTQW). We trotterize a CTQW to recover a QAOA sequence, and employ recent advances on the theory of Trotter formulas to bound the query complexity (circuit depth) needed to prepare a state that approaches perfect overlap with the target state. The obtained complexity exceeds the Grover's algorithm complexity $O\left(N^\frac{1}{2}\right)$, but remains smaller than $O \left(N^{\frac{1}{2}+c}\right)$ for any $c>0$, which shows quantum advantage of QAOA over classical solutions. We verify our analytical predictions by numerical simulations of up to 68 qubits.
\end{abstract}

 \maketitle

\section {Introduction}
Variational quantum algorithms (VQAs) have become the de facto model of quantum computation of the noisy intermediate scale quantum (NISQ) computers era. These algorithms make use of a parameterized quantum circuit (a.k.a ansatz) whose parameters are iteratively tuned by a classical co-processor to minimize a cost function---a procedure that has been shown to alleviate some of the limitations of NISQ computers \cite{sharma2020noise, gentini2020noise, cincio2021machine}. 
A prominent example of a variational algorithm, called quantum approximate optimization algorithm (QAOA) was originally designed to approximate solutions to combinatorial optimization problems \cite{farhi2014quantum}.
It was later shown to be effective for a variety of other tasks, and is currently one of the most studied quantum algorithms. 
Milestones of QAOA research include experimental realizations using 23 qubits \cite{harrigan2021quantum}, several results that aid and improve on the original implementation of the algorithm \cite{Zhou2020,wang2020x,Brady2021,akshay2021parameter}, and universality results \cite{lloyd2018quantum,morales2020universality}, implying the importance of QAOA beyond the NISQ era.
Among the promising uses of QAOA is its application to the problem of unstructured search, where the objective is to prepare a particular target bit string. Notable results related to unstructured search via QAOA include: the discovery of parameter concentrations \cite{akshay2021parameter}, optimal depth scaling when using a modified mixer \cite{morales2018variational}, and near optimal depth scaling when preparing a state with an overlap of $\frac{1}{2}$ with the target state \cite{jiang2017near}. { Nevertheless, the depth scaling for approaching perfect overlap, and how it compares to the complexity of Grover's algorithm \cite{grover1996fast} were unknown.}

The problem of unstructured search has been studied in the context of many other models of quantum computing including adiabatic quantum computing~\cite{farhi2000quantum} and continuous time quantum walks (CTQWs)~\cite{childs2004spatial}.
CTQW is a type of quantum algorithm used to simulate the dynamics of a quantum system given a graph that describes the transitions between states. In this context, ``continuous" refers to the fact that the transitions from one state to another occur continuously over time rather than at discrete steps as in a discrete time quantum walk. Originally proposed by Farhi and Gutmann \cite{farhi1998quantum}, it relates to a classical continuous time random walk by the analogy between a classical transfer matrix and a quantum  Hamiltonian. 
CTQWs are known to provide an exponential speedup in certain problems~\cite{childs2003exponential}, recover Grover's search~\cite{childs2004spatial}, and be a computationally universal model~\cite{childs2009universal}.

Continuous time evolution in CTQWs can be approximated by discretizing them into a quantum circuit with the use of product formulas, like the Suzuki-Trotter formulas \cite{berry2007efficient,suzuki1976generalized}. Notably, CTQW solving the unstructured search problem recovers a QAOA sequence upon this discretization. 

In this paper we make use of recent developments in the analysis of product formula errors \cite{childs2021theory} to discretize CTQW for the problem of unstructured search, which recovers a QAOA sequence. As a result, we identify an upper bound on the QAOA circuit depth (query complexity) for unstructured search, sufficient to approach perfect overlap. Our complexity bound exceeds the well known Grover complexity $O(N^\frac{1}{2})$ but scales slower than $O \left(N^{\frac{1}{2}+c}\right)$ for any $c>0$. Our analysis is supported by numerical evidence up to 68 qubits. Compared to the results presented by Jiang et al.~\cite{jiang2017near}, our method offers an improvement to the overlap in exchange for a slightly higher complexity.
This result sheds light onto the power of QAOA and makes a step towards closing the gap between QAOA and Grover’s algorithm in terms of scaling and overlap.

\section{Preliminaries}

\subsection{Unstructured search via CTQW}

In the most general sense, a continuous time quantum walk is described by a Hamiltonian $H$ which induces a unitary evolution, defined by the operator $e^{-iHt}$. Here Hamiltonian $H$ is analogous to the transfer matrix of classical continuous time random walks.

In this setting, solving unstructured search in a hypercube graph consists of the evolution under the operator
\begin{equation}
    U(\alpha,t)=e^{-i(\alpha H_x + H_w)t},
    \label{CTQW_ansatz}    
\end{equation}
where 
$H_x=\sum_{j=1}^n X_{j}$ is the adjacency matrix of the $n$-dimensional hypercube graph, with $X_j$ being a Pauli $X$ matrix applied to the $j$-th qubit;  $H_w=\ketbra{w}{w}$ is the projector on the target state $\ket{w} \in \{\ket{0},\ket{1}\}^{\otimes n}$, 
% $t$ is time, and 
$\alpha$ is a tunable parameter, and $t$ is the evolution time.
The objective is then to maximize the overlap
\begin{equation}
\label{ctqw_overlap}
    \max_{\alpha, t}|\bra{w}  U(\alpha,t) \ket{+}^{\otimes n}|^2.
\end{equation}
We can see that \eqref{ctqw_overlap} is target independent by substituting $\ket{w}=U_x \ket{0}^{\otimes n}$, where $U_x \in \{X, \mathbb{1}\}^{\otimes n}$ \cite{campos2021training}.
Therefore, without loss of generality we set  $\ket{w}=\ket{0}^{\otimes n}$ and $H_w=H_0=(\ketbra{0}{0})^{\otimes n}$.
The optimal value of $\alpha$ was calculated by Farhi~et~al.~\cite{farhi2000quantum} to be 
\begin{equation}
    \alpha^*=\frac{1}{2^n}\sum_{k=1}^{n} \frac{C_k^n}{k}=\frac{1}{n}+O\left(\frac{1}{n^2}\right),
    \label{alpha}
\end{equation}
where $C_k^n$ are binomial coefficients (details in Appendix \ref{ap:alpha_delta}).
% \todo[inline]{Though the exact value instead of $O$ is of crucial importance.}
The evolution $U(\alpha^*,t)\ket{+}^{\otimes n}$ largely occurs in the two-dimensional subspace spanned by the low-energy eigenstates of $\alpha^* H_x + H_0$ (see Appendix \ref{ap:eigenstates}). These two eigenstates are approximately given by
\begin{align}
\ket{\psi_+}= \frac{1}{\sqrt{2}}(\ket{+}^{\otimes n}+\ket{0}^{\otimes n})+O\left(\frac{1}{n}\right),\label{ctqw_eigenstates}\\
\ket{\psi_-}=\frac{1}{\sqrt{2}}(\ket{+}^{\otimes n}-\ket{0}^{\otimes n})+O\left(\frac{1}{n}\right),\notag
\end{align}
with an energy gap $\Delta = \frac{2}{\sqrt{2^n}}\left(1+O\left(\frac{1}{n}\right)\right)$ \cite{childs2002quantum} (details in Appendix \ref{ap:eigenstates}).
Thus, up to a global phase we can approximate the evolution as
\begin{equation}
\label{eq:approx_u_plus}
    U(\alpha^*,t)\ket{+}^{\otimes n}=\frac{1}{\sqrt{2}}\left(\ket{\psi_+}+e^{-i\Delta t}\ket{\psi_-}\right)+O\left(\frac{1}{n}\right),
\end{equation}
allowing one to establish that for $t^*=\frac{\pi}{2}\sqrt{2^n}$ the overlap becomes
\begin{equation}
    |\bra{0}^{\otimes n}U(\alpha^*,t^*)\ket{+}^{\otimes n}|^2 = 1 + O\left(\frac{1}{n}\right).\label{eq:CTQW_max_ovl}
\end{equation}

\subsection{Unstructured search via QAOA}
An alternative way of solving the unstructured search problem is by employing QAOA. With $\ket{w}$ representing a target bit string in the computational basis, the task is to variationaly prepare a candidate state of high overlap with $\ket{w}$. In QAOA, an $n$ qubit candidate state $\ket{\psi_p (\bm\gamma, \bm\beta)}$---prepared by a circuit of depth $p$---is parametrized as:
\begin{equation}
    \ket{\psi_p(\bm\gamma,\bm\beta)} =  \prod\limits_{k=1}^p e^{-i \beta_k H_{x}} e^{-i \gamma_k H_w}\ket{+}^{\otimes{n}},
    \label{ansatz}
\end{equation} 
with real parameters $\gamma_k\in[0,2\pi)$, $\beta_k\in[0,\pi)$. 

The optimization task is to maximize the overlap between the candidate state $\ket{\psi_p(\bm\gamma,\bm\beta)}$ and the target state $\ket{w}$ given by
\begin{equation}
0\leq \max_{\bm \gamma, \bm \beta}|\braket{w}{\psi_p(\bm \gamma, \bm \beta)}|^2\leq 1
\label{overlap}
\end{equation}
Similar to CTQW, without loss of generality, one can set $\ket{w}=\ket{0}^{\otimes n}$.

\subsection{Product formulas}
Despite a similar formulation, the two discussed approaches for the unstructured search problem are essentially different. Indeed, QAOA presents a discrete evolution with the Hamiltonians $H_w$ and $H_x$ exponentiated individually, as in \eqref{ansatz}, while CTQW represents a continues evolution with both Hamiltonians appearing in the same exponent, as in \eqref{CTQW_ansatz}. Nevertheless, the so called product formulas can relate the two approaches, by providing an approximation to an operator of the form $U(t)=e^{t\sum_{\mu =1}^M H_\mu}$ with a product of exponentials of individual $H_\mu$ \cite{berry2007efficient}. The precision of that approximation depends on the so-called order of the formula, which controls number of individual $H_\mu$ exponentiations in the sequence. 
High order Trotter-Suzuki formulas $S_q(t)$ \cite{suzuki1976generalized} can be generated recursively as 
\begin{align}
    % &S_1(t)=e^{H_M t}\cdots e^{H_1t},\\
    &S_2(t)=e^{H_M t/2}\cdots e^{H_1 t/2}e^{H_1 t/2}\cdots e^{H_M t/2},\\
    &S_{2k}(t)=S_{2k-2}^2(u_k t)S_{2k-2}((1-4u_k) t)S_{2k-2}^2(u_k t),\label{suzuki_2}
\end{align}
where $u_k=1/(4-4^{1/(2k-1)})$. In general, a $q$ order product formula can be written in the form 
% \todo{change that $M$ so it is not confusing}
\begin{equation}
    S_{q}(t)=\prod_{v=1}^{\Upsilon_q}\prod_{\mu=1}^M e^{ta(v,\mu)H_{\pi(v,\mu)}},
\end{equation}
where $\Upsilon_q$ is the number of stages of the formula.
The formulas of $q$th order approximate the original operator $U(t)$ with a sequence of individual $H_\mu$ exponentials with an error of $O(t^{q+1})$.

In order to approximate an evolution $U(t)$ with large $t$ one should partition an evolution into $r$ smaller Trotter steps $S_q(t/r)$. As calculated by Childs et al.~in  \cite{childs2021theory}, for unitary $U(t)$ this results in an approximation error 
\begin{equation}
    \|U(t)-S^r_q(t/r)\|_2\leq \epsilon = \frac{2(\Upsilon_q)^{q+1}\delta (t^*)^{q+1}}{r^q (q+1)},
    \label{error}
\end{equation}
where, 
\begin{equation}
    \delta=\sum_{\mu_1,\mu_2,\cdots\mu_{q+1}=1}^M \|[H_{\mu_{q+1}},\cdots [H_{\mu_2},H_{\mu_{1}}]\cdots]\|_2.
    \label{delta_nest}
\end{equation}

\section{Depth scaling of QAOA from a trotterized CTQW}

\begin{theorem}[Complexity of search via QAOA]\label{complexity}
{ A QAOA circuit of depth 
\begin{equation}
    p=O\left(2^{\frac{n}{2}+\sqrt{n}\sqrt{2\cdot\log_2 5}}\right)
    \label{theorem}
\end{equation}
can prepare a state that satisfies $|\braket{\omega}{\psi_p(\bm\gamma,\bm\beta)}|^2 = 1 + O(1/n)$.}

\end{theorem}
\begin{proof}
The evolution operator $U(\alpha^*,t^*)$ can be approximated by a $q$th order product formula 
\begin{equation}
    S_{q}^r(t^*/r)=\prod_{v=1}^{r\Upsilon_q} e^{t^*b(v)H_2} e^{t^*a(v)H_1},\label{eq:formula_form}
\end{equation}
where $a(v),b(v) \in \mathbb{R}$, $H_1=-iH_0$, $H_2=-iH_x \alpha^*$, and $\Upsilon_q = 5^{q/2-1}$ (details in Appendix \ref{appendix:CTQW_to_QAOA}). Equation \eqref{eq:formula_form} gives a QAOA sequence of depth $p=r\Upsilon_q$. The objective is to calculate the depth that results in an approximation with an error $\|U(\alpha^*,t^*)-S^r_q(t/r)\|_2 \leq \epsilon$. This can be calculated from solving \eqref{error} for $p$,
\begin{equation}
p=\frac{(\Upsilon_q)^{2+1/q}(2\delta)^{1/q} (t^*)^{1+1/q}}{(\epsilon(q+1))^{1/q}}.
\label{formula_depth}
\end{equation}

We make use of the fact that in unstructured search the evolution happens in a symmetric subspace of dimension $n+1$. 
\begin{definition}[Symmetric subspace]
$H_s=Span\{\ket{\psi}: P_{ij}\ket{\psi}=\ket{\psi}\}$, where $P_{ij}$ is a permutation of arbitrary qubits $i$ and $j$.  
\end{definition}

\begin{definition}[Dicke basis vectors]\label{def:dicke_basis}
$\ket{e_k} = \dfrac{1}{\sqrt{C_k^n}}\sum\limits_{z_1+\cdots+z_n=k} \ket{z_1\cdots z_n}$, where $z_i\in\{0,1\}$, and $C_k^n$ are binomial coefficients.
\end{definition}

In this subspace, the operator $H_x$ is a tridiagonal matrix with diagonal elements equal to 0, and off-diagonal elements given by 
\begin{equation}
    \bra{e_{l+1}}H_x\ket{e_l}=\sqrt{(l+1)(n-l)}\leq \frac{n+1}{2}.
\end{equation}
It follows that the max norm  $\|A\|_{\max} = \max_{j,k} |A_{j,k}|$
of $H_2=-iH_x \alpha^*$ can be upper bounded as
\begin{equation}
    \|H_2\|_{\max}\leq \alpha^*\frac{n+1}{2}.
\end{equation}
As for $H_0$, in the symmetric subspace it takes the form of an $n+1$ dimensional square matrix with all elements equal to 0 except $\bra{e_0}H_0\ket{e_0}=1$.

Notice that for an arbitrary matrix $A$ it follows that,
\begin{equation}
    \|[H_1,A]\|_{\max}\leq \|A\|_{\max},\label{max_H1}
\end{equation}
\begin{align}
    \|[H_2,A]\|_{\max} \leq\max_{j,k} &\bigg(\Big(|A_{j+1,k}|+|A_{j-1,k}|+\notag\\
    +|A_{j,k+1}|+&|A_{j,k-1}|\Big)\left(\alpha^*\frac{n+1}{2}\right)\bigg)\notag\\
    \leq & 2\alpha^*(n+1)\|A\|_{\max}. \label{max_H2}
\end{align}
Thus, in order to find an upper bound for the norm of nested commutators in \eqref{delta_nest} we focus only on how many times $H_2$ appears in the sequence. The largest matrix element resulting from any such sequence is bounded by 
\begin{equation}
    \|[H_{\mu_{q+1}},\cdots [H_{\mu_2},H_{\mu_{1}}]\cdots]\|_{\max} \leq  (2\alpha^*\left(n+1\right))^{j}.\label{max_nested}
\end{equation}
where $\mu_1\neq \mu_2$, and $j$ being the number of times $H_2$ appears in the sequence.

We recall the inequality between the spectral and Frobenius  norm $\| \cdot \|_\text{F}$ as,
\begin {equation}
\|A\|_2 = \sigma_{\max}(A)\leq\|A\|_{\rm F} = \left(\sum_{j=1}^m \sum_{k=1}^n |A_{j,k}|^2\right)^{\frac{1}{2}},\label{f_norm}
\end{equation}
where $\sigma_{\max}(A)$ is the largest singular value of $A$.

Using \eqref{max_nested} and \eqref{f_norm} we find the following expression for the spectral norm of nested commutators
    \begin{equation}
       \|[H_{\mu_{q+1}},\cdots [H_{\mu_2},H_{\mu_{1}}]\cdots]\|_2\leq \left(2\alpha^*\left(n+1\right)\right)^{j}(n+1), \label{nested_norm}
    \end{equation}
where $\mu_1\neq \mu_2$, and $j$ is the number of times $H_2$ appears in the sequence.

Substituting \eqref{nested_norm} into the definition of $\delta$ \eqref{delta_nest} we obtain 
\begin{eqnarray}
    \delta\leq 2(n+1)\sum_{j=1}^q  \left(2\alpha^*(n+1)\right)^{j} C^{q-1}_{j-1}
    \label{delta_2}
    \\
    \leq 2(n+1)\left(2\alpha^*(n+1)+1\right)^q).\label{delta_final}
\end{eqnarray}
The binomial coefficients in \eqref{delta_2} come from the number of combinations in which $H_2$ can appear $j-1$ times outside the first commutator. An extra factor of $2$ appears from the first commutator being either $[H_1,H_2]$ or $[H_2,H_1]$.

After substituting $\Upsilon_q=5^{q/2-1}$, $t^*=\frac{\pi}{2}\sqrt{2^n}$, and \eqref{delta_final} into \eqref{formula_depth}, we derive
\begin{equation}
    p \leq p_0 \sqrt{2^n} \left(\frac{2\pi (n+1)\sqrt{2^n}}{5\epsilon}\right)^{\frac{1}{q}}5^q,
    \label{depth_qaoa_qtqw}
\end{equation}
where 
\begin{equation}
    p_0=\frac{\pi \left(2\alpha^*(n+1)+1\right)}{2\cdot 5^{3/2}(q+1)^{\frac{1}{q}}}.
    \label{p_0}
\end{equation}
To minimize \eqref{depth_qaoa_qtqw} we take its derivative and set it to zero. Neglecting a small negative contribution, we find the value of $q$ that results in the shortest depth to be
\begin{equation}
    q\leq\left( \frac{n\cdot \ln\sqrt{2}+\ln(2\pi(n+1))-\ln5\epsilon} {\ln 5}\right)^\frac{1}{2}. \label{order_q}
\end{equation}
Substituting \eqref{order_q} into \eqref{depth_qaoa_qtqw} and simplifying we arrive at
\begin{align}
    p \leq p_0\left(\frac{2\pi(n+1)}{5\epsilon}\right)^{\frac{2}{q}}
    2^{\frac{n}{2}+\sqrt{n}\sqrt{2\log_2 5}}. \label{depth_qaoa_ctqw_subs}
\end{align}
Importantly, from \eqref{p_0} and \eqref{alpha} one can show that $p_0<\frac{1}{2}$. { Moreover, since the CTQW is limited to an overlap $1+O(1/n)$, as given by \eqref{eq:CTQW_max_ovl}, errors smaller than $\epsilon = O(1/n)$ do not meaningfully alter the overlap of QAOA. Therefore, for errors of this order and large $n$, the second factor in \eqref{depth_qaoa_ctqw_subs} does not grow with $n$. This leaves the last factor in \eqref{depth_qaoa_ctqw_subs} as the dominant contribution, establishing \eqref{theorem}}.

\end{proof}

\begin{remark}[Complexity comparison]
    Theorem \ref{complexity} puts the query complexity of  the derived QAOA circuit in between the traditional Grover's algorithm and classical search
\begin{equation}
    O\left(N^{\frac{1}{2}}\right)< O\left(N^{\frac{1}{2}}2^{\sqrt{\log_2 N}\sqrt{2\cdot\log_2 5}}\right)< O(N),
\end{equation}
where $N=2^n$.
\end{remark}
It is also worth noting that 
\begin{equation}
    O\left(N^{\frac{1}{2}}2^{\sqrt{\log_2 N}\sqrt{2\cdot\log_2 5}}\right)< O \left(N^{\frac{1}{2}+c}\right),
\end{equation}
for any $c>0$.

QAOA angles can be recovered by expressing $ S^r_q(t^*/r)$ as a product of second order terms by following \eqref{suzuki_2} and recursively calculating the parameters $t_k$ that correspond to each second order term as,
\begin{equation}
    S^r_q(t^*/r) = \prod_{k=1}^{p} S_2(t_k).
\end{equation}
Then, QAOA angles can be recovered by
\begin{equation}
    \beta_k=t_k,~ \gamma_{k\neq p}=\frac{t_k+t_{k+1}}{2}, ~ \gamma_p=\frac{t_p}{2}.
\end{equation}

\section{Numerical experiments}
We begin testing our analytics by confirming that the resulting QAOA sequence has sufficient depth to precisely approximate the evolution of the corresponding CTQW.
{ We compare (i) Grover search, (ii) QAOA from a trotterized CTQW with numerically calculated depth, and (iii) QAOA from a trotterized CTQW with analytically predicted depth. For these three cases, Figure \ref{fig:troterized} illustrates the overlaps increasing through the sequences for $\epsilon = 0.01$ and $n = 42, 46$.
Here $n=46$ is the largest system size with the optimal order $q=6$, allowing for the numerical estimation of the optimal depth with high precision (details on the numerical execution of these circuits and search for optimal depths can be found in Appendix \ref{appendix:numerical}). 
For the QAOA sequence obtained analytically we use a $q=4$ order formula, as calculated from equation \eqref{order_q}.
We observe the overlaps from the QAOA sequences to follow smooth curves reminiscent of those of CTQWs with respect to $t$, reaching $1 + O\left(\frac{1}{n}\right)$ as in equation \eqref{eq:CTQW_max_ovl}.}   

\begin{figure}
    \centering
    \includegraphics[width=0.5\textwidth]{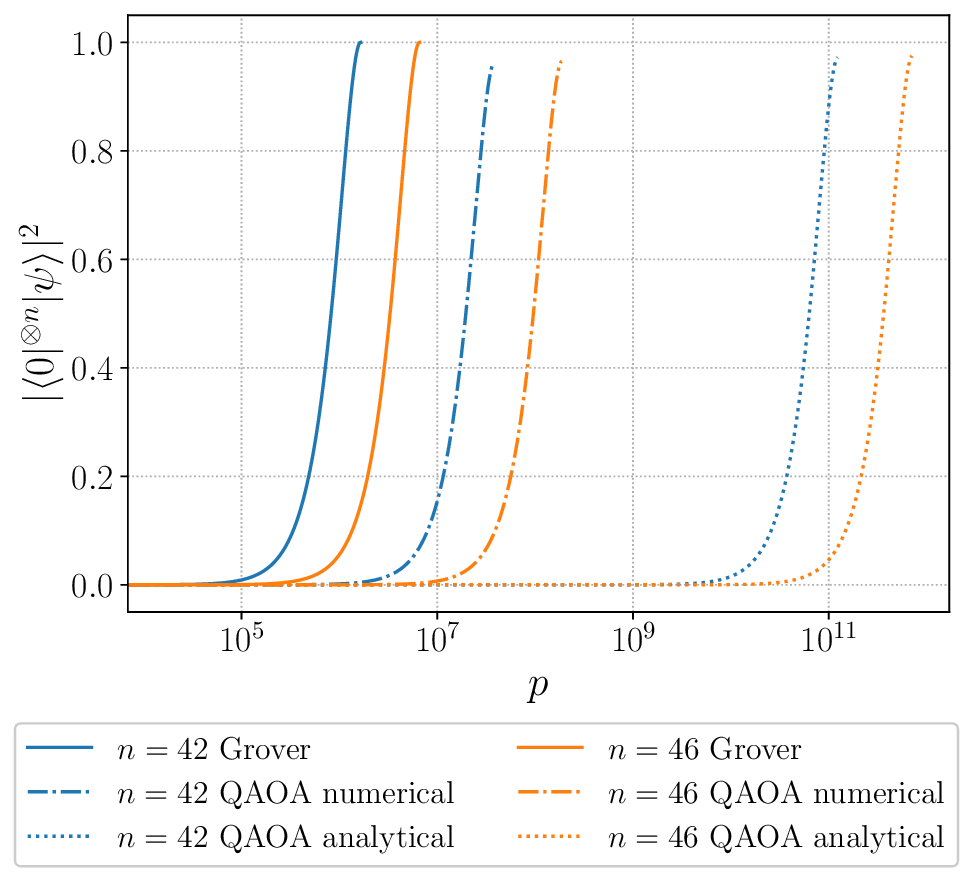}
    \caption{{Overlap with the target state $\ket{0}^{\otimes n}$ through the circuit for states $\ket{\psi}$ prepared by: (i) Grover search, (ii) QAOA from a trotterized CTQW with numerically calculated depth, and (iii) QAOA from a trotterized CTQW with analytically predicted depth for $n=42,46$ and $\epsilon= 0.01$. In agreement with \eqref{eq:CTQW_max_ovl}, the overlaps of the QAOA sequences reach $1 + O\left(\frac{1}{n}\right)$. }}
    \label{fig:troterized}
\end{figure}

In order to verify the tightness of expression \eqref{depth_qaoa_ctqw_subs}, we exhaustively calculate the depth required by a trotterized CTQW to have an error below certain threshold $\epsilon$. Figure \ref{fig:depth_ratio} illustrates the ratio between the depth predicted by equation \eqref{depth_qaoa_ctqw_subs} and the depth calculated numerically in the range { $n \in [22,68]$} and $\epsilon=\{ 0.001, 0.01, 0.1\}$. When numerically approximating the optimal depth, for each pair $n$ and $\epsilon$, we use the formula of the order which results in the shortest sequence. The sharp fluctuations that appear at $q=8$ are likely due to the longer sequences of every step $S_8(t/r)$, and the more computationally intensive task of approximating depths for the larger system sizes. Across the entire range the ratio demonstrates a seven fold increase, showing a much slower growth compared to the dominant factor $2^\frac{n}{2}$.

\begin{figure}
    \centering
    \includegraphics[width=0.5\textwidth]{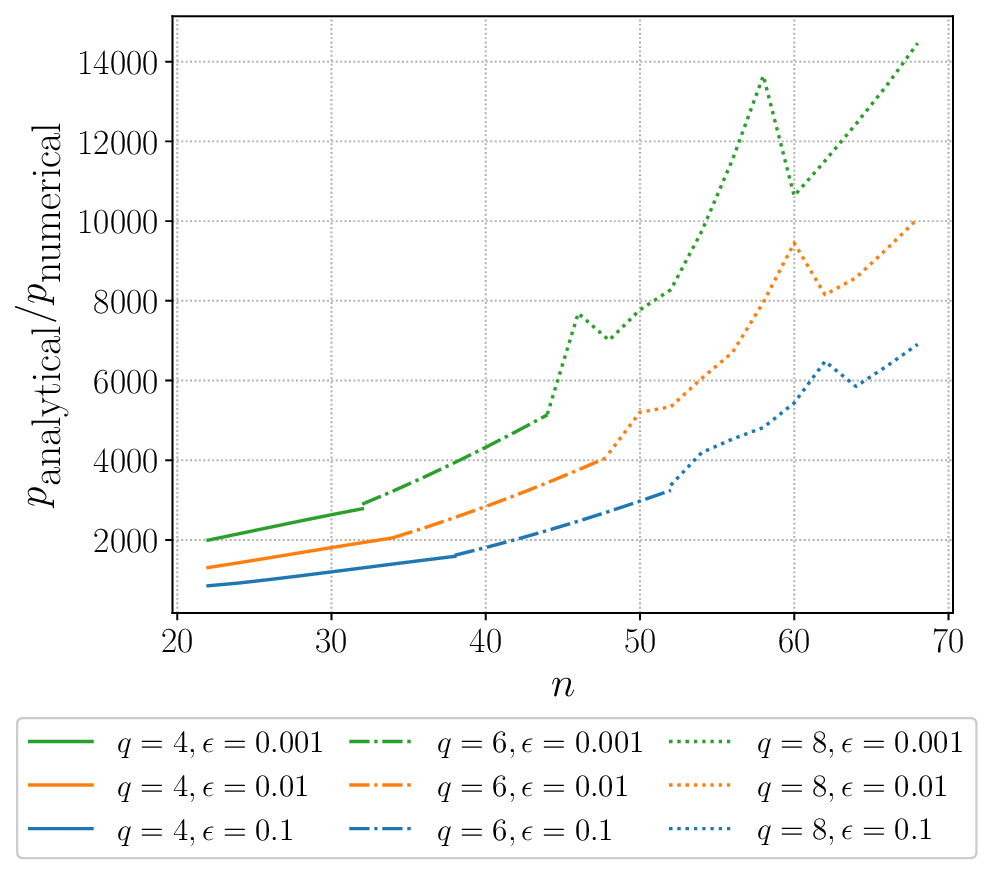}
    \caption{Ratios between the analytically calculated depth $p_\mathrm{analytical}$, given by \eqref{depth_qaoa_ctqw_subs}, and numerically calculated depth $p_\mathrm{numerical}$ for system sizes $n \in [22,68]$ and $\epsilon\in \{ 0.001, 0.01, 0.1\}$. For each pair $n$ and $\epsilon$, $p_\mathrm{numerical}$ is calculated using the formula of order $q$ which results in the shortest sequence.}
    \label{fig:depth_ratio}
\end{figure}

Figure \ref{fig:depth_eps} illustrates numerically calculated depth with respect to $\epsilon\in[0.001, 0.1]$ for $n=40,42,44$ and $q=6$. It can be observed that depth growth slows down as $\epsilon$ becomes smaller. In equation \eqref{depth_qaoa_ctqw_subs}, the impact of $\epsilon$ on depth is given by the factor $\epsilon^{-\frac{2}{q}}$ which diminishes for higher order formulas (large systems), and goes to 1 in the limit $n\rightarrow \infty$.

\begin{figure}
    \centering
    \includegraphics[width=0.5\textwidth]{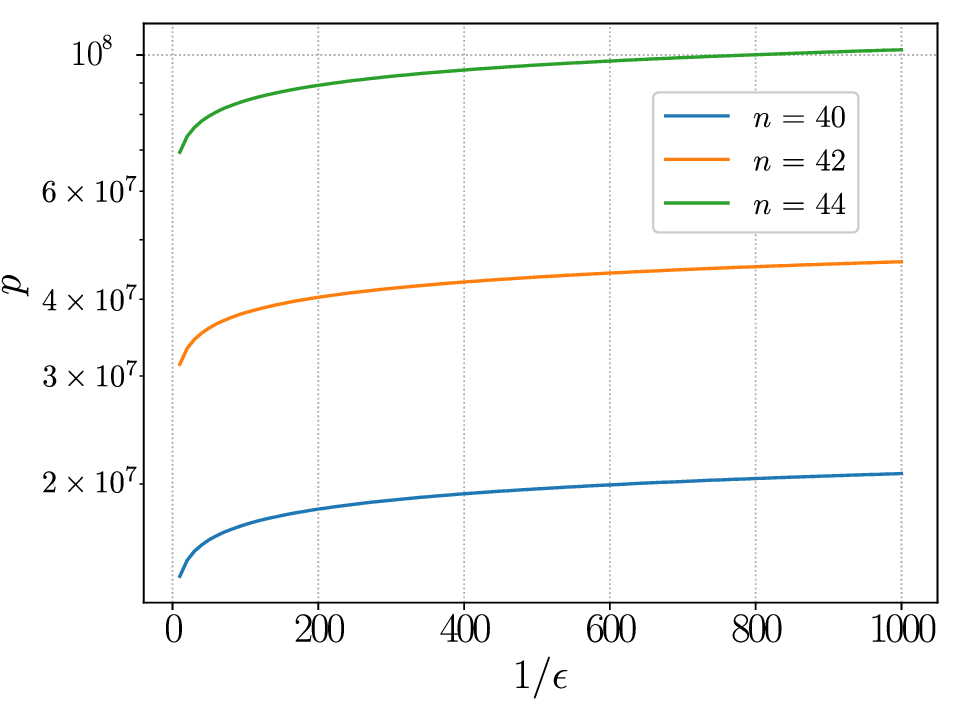}
    \caption{Numerically calculated depths for a range of $\epsilon\in [0.001,0.1]$ for $n=40, 42, 44$ and $q=6$.}
    \label{fig:depth_eps}
\end{figure}

\section{Discussion}
Our results demonstrate quantum advantage for unstructured search via QAOA where the ansatz state approaches perfect overlap with the target. We prove the query complexity to be higher than that of Grover's algorithm $O(N^{\frac{1}{2}})$, but notably smaller than $O\left(N^{\frac{1}{2}+c}\right)$ for any $c>0$.
Compared to the results presented by Jiang et al. \cite{jiang2017near}, our approach offers an improvement to the overlap in exchange for a higher complexity.

It might be tempting to tighten the bound \eqref{depth_qaoa_ctqw_subs} by considering evolution in a two dimensional subspace spanned by the low energy eigenstates \eqref{ctqw_true_eigenstates}. However, the estimate obtained this way does not match the numerical experiments, predicting depths of the order of $O(2^\frac{n}{2})$, lower than what we numerically observe. This discrepancy is likely caused by a probability leakage out of this two dimensional subspace.

Nevertheless, it may still be possible to improve the derived upper bound for depth.
    Equation \eqref{error}, as introduced by Childs \cite{childs2021theory}, makes use of multiple upper bounds in an attempt to make it as general as possible, which opens the door for finding a tighter expression tailored to our particular setting. Similarly, in our analysis we make use of some generous upper bounds in order to obtain a closed form of $\delta$, which offers room for improvement. { Regardless of this, it is unlikely that any improvement would lower the bound to the point of optimal scaling. As we observe numerically, the suboptimal complexity appears to be a fundamental feature of the trotterization of the considered CTQW. Alternatively, one may consider the trotterization of hybrid adiabatic and CTQW evolutions, similar to those presented in \cite{morley2019quantum}. Such evolutions also present optimal scaling and may result in Trotter sequences of lower complexity.}

{ Our result provides insight on the power of QAOA and makes a step towards closing the gap between QAOA and Grover's algorithm in terms of scaling and overlap. If this gap were to be closed, QAOA could become a more attractive alternative to Grover's algorithm due to the easier implementation of the mixer $H_x$ compared to the projector $\ketbra{+}{+}^{\otimes n}$.}

 \section{Acknowledgments}

The work of E.C. and D.R. was supported by Rosatom in the framework of the Roadmap for Quantum computing (Contract No. 868-1.3-15/15-2021 dated October 5, 2021 and Contract No. №R2163 dated December 03, 2021).

\onecolumngrid
\newpage

\appendix

\section{Optimal value of $\alpha$ and energy gap}
\label{ap:alpha_delta}
For completeness here we reiterate the calculation of optimal value of $\alpha$ for CTQW, presented in \cite{farhi2000quantum}.
Let $H= \alpha H_x + \ketbra{\mathbf{0}}{\mathbf{0}}$, where $\ket{\mathbf{0}}\equiv\ket{0}^{\otimes n}$. We seek to find $\alpha$ for which the energy gap $\Delta$ is minimum.

Let $\ket{e_k}$ be the Dicke states as in Definition \ref{def:dicke_basis}, and $\ket{h_k}=H_a ^{\otimes n} \ket{e_k}$ where $H_a$ is the Hadamard gate. It can be seen that
\begin{equation}
    \left(\sum_{i=1}^n Z_i\right) \ket{e_k}= (n-2k)\ket{e_k},
\end{equation}
\begin{equation}
    H_x\ket{h_k}=(n-2k)\ket{h_k}.
    \label{eq:H_x_eigen}
\end{equation}

We solve for the eigenvalues of $H$, $H\ket{\psi}=E\ket{\psi}$. Multiplying it by $\bra{h_k}$, 
\begin{equation}
    \bra{h_k}H\ket{\psi} = E \braket{h_k}{\psi},
\end{equation}
substituting $H$ and using property \eqref{eq:H_x_eigen} after a series of algebraic manipulations we end up with \eqref{eq:hk_psi}:
\begin{equation}
    \alpha(n-2k)\braket{h_k}{\psi}+\braket{h_k}{\mathbf{0}}\braket{\mathbf{0}}{\psi} = E \braket{h_k}{\psi},
\end{equation}
\begin{equation}
    [E-\alpha(n-2k)]  \braket{h_k}{\psi} =\braket{h_k}{\mathbf{0}}\braket{\mathbf{0}}{\psi},
\end{equation}
\begin{equation}
    \braket{h_k}{\psi} =\frac{\braket{h_k}{\mathbf{0}}\braket{\mathbf{0}}{\psi}}{E-\alpha(n-2k)}.\label{eq:hk_psi}
\end{equation}
We multiply both sides by $\braket{\mathbf{0}}{h_k}$ and sum over $k$,
\begin{equation}
    \sum_{k=0}^n\braket{\mathbf{0}}{h_k}\braket{h_k}{\psi} =\sum_{k=0}^n\frac{|\braket{h_k}{\mathbf{0}}|^2\braket{\mathbf{0}}{\psi}}{E-\alpha(n-2k)},
\end{equation}
\begin{equation}
    1 =\sum_{k=0}^n\frac{P_k}{E-\alpha(n-2k)},
\end{equation}
where we used the fact that $\ket{h_k}$ form a basis in the symmetric subspace. Here $P_k= |\braket{h_k}{o}|^2=\frac{C_n^k}{2^n}$.
We introduce a change of variables
\begin{equation}
   \lambda = E/\alpha \implies \alpha = \sum_{k=0}^n\frac{P_k}{\lambda-(n-2k)}. \label{eq_alpha}
\end{equation}
Let $\alpha >0$, the right hand side of \eqref{eq_alpha} approaches $+\infty$ as $\lambda \rightarrow n+0, n-2+0, \cdots -n+0$, and goes to $-\infty$ as $\lambda \rightarrow n-0, n-2-0, \cdots -n-0$.
We will prove that there exist two roots exponentially close to $\lambda=n$ for
\begin{equation}
    \alpha ^*=\frac{1}{2}\sum_{k=1}^n\frac{P_k}{k}. \label{eq:alpha*}
\end{equation} 
Substituting $\alpha^*$ from \eqref{eq:alpha*} to \eqref{eq_alpha}, we get
\begin{equation}
    \frac{1}{2}\sum_{k=1}^n\frac{P_k}{k}=\sum_{k=0}^n\frac{P_k}{\lambda+2k-n}\iff \sum_{k=1}^n\frac{P_k(\lambda-n)}{2k(\lambda+2k-n)}=\frac{P_0}{\lambda - n}=\frac{1}{2^n(\lambda-n)} \label{eq:roots1}.
\end{equation}
As we search for $|\lambda-n|<<1$, we neglect the term in the denominator of the LHS of \eqref{eq:roots1},
\begin{equation}
    \frac{1}{4}\sum_{k=1}^n\frac{P_k}{k^2}=\frac{1}{2^n(\lambda-n)^2}\implies\lambda = n \pm \xi,
\end{equation}
where $\xi=\frac{2}{\sqrt{2^n}}\left(\sum_{k=1}^n\frac{P_k}{k^2}\right)^{-\frac{1}{2}}$. The energy gap is then 
\begin{equation}
    \Delta=2\alpha^*\xi,
    \label{eq:gap_pr}
\end{equation} which is correct up to exponential precision $O(2^{-n})$. However, to obtain a tractable expression we approximate sums in the expression for $\alpha^*$ and $\xi$
\begin{equation}
    \sum_{k=0}^n\frac{P_k}{k}= \frac{2}{n}+O\left(\frac{1}{n^2}\right), ~  \sum_{k=0}^n\frac{P_k}{k^2}= \frac{4}{n^2}+O\left(\frac{1}{n^3}\right),
\end{equation}
allowing us to conclude
\begin{equation}
    \implies \Delta = \frac{2}{\sqrt{2^n}}\left(1+O\left(\frac{1}{n}\right)\right).
    \label{eq:gap}
\end{equation}

\section{Low energy eigenstates of $H$}
\label{ap:eigenstates}
From equations \eqref{eq:hk_psi} and \eqref{eq:alpha*} we have, 
\begin{equation}
\alpha^*\braket{h_k}{\psi}=\frac{\sqrt{P_k}\braket{\mathbf{0}}{\psi}}{\lambda-n+2k},
\end{equation}
which for $k=0$ simplifies to
\begin{equation}
    \alpha^*\braket{h_0}{\psi}=\pm\frac{\sqrt{P_0}}{\xi}\braket{\mathbf{0}}{\psi}.\label{eq:alpha+psi}
\end{equation}
Note that $\ket{h_0}=\ket{+}^{\otimes n}$. For $k\neq 0$ exponentially small $\lambda-n$ can be neglected, giving
\begin{equation}
    \alpha^*\braket{h_k}{\psi}=\frac{\sqrt{P_k}}{2k}\braket{\mathbf{0}}{\psi}.\label{eq:alphahkpsi}
\end{equation}
\begin{equation}
    \sum_{k=0}^n|\braket{h_k}{\psi}|^2=1=\frac{|\braket{\mathbf{0}}{\psi}|^2}{(\alpha^*)^2}\left(\frac{P_0}{\xi^2}+\sum_{k=1}^n\frac{P_k}{4k^2}\right)
\end{equation}
Thus, 
\begin{equation}
   (\alpha^*)^2 =|\braket{\mathbf{0}}{\psi}|^2\left(\frac{P_0}{\xi^2}+\sum_{k=1}^n\frac{P_k}{4k^2}\right) = |\braket{\mathbf{0}}{\psi}|^2\left(\frac{2}{2^n\xi^2}\right)
\end{equation}
Where we used equation \eqref{eq:roots1} at the last transition. Finally, we conclude

% \begin{equation}
%     |\braket{\mathbf{0}}{\psi}|^2=\frac{2^n\Delta^2}{8}
% \end{equation}

\begin{equation}
    \braket{\mathbf{0}}{\psi_\pm}=\pm\frac{\sqrt{2^n}\Delta}{2\sqrt{2}} 
    % =\frac{1}{\sqrt{2}} +O\left(\dfrac{1}{n}\right)
    .
    \label{eq:0psi_pm}
\end{equation}
% Notice that polynomial precision appears here only from using equation for the gap \eqref{eq:gap}. Keeping the 
From equations \eqref{eq:alpha+psi} and \eqref{eq:0psi_pm}
\begin{equation}
     \braket{h_0}{\psi_\pm}=\frac{1}{\sqrt{2}}+O(2^{-n}),\label{eq:+psipm}
\end{equation}
Similarly using \eqref{eq:alphahkpsi}, \eqref{eq:0psi_pm} for $k\neq0$,
\begin{equation}
    \braket{h_k}{\psi_\pm}=\pm\frac{\sqrt{P_k}}{k\sqrt{2}}\left(\sum_{k=0}^n\frac{P_k}{k^2}\right)^{-\frac{1}{2}}+O(2^{-n}), \label{eq:hkpsipm}
\end{equation}
From equations \eqref{eq:+psipm} and \eqref{eq:hkpsipm} we obtain the eigenstates
\begin{equation}
    \ket{\psi_\pm}= \frac{1}{\sqrt{2}}\left(\ket{+}^{\otimes n}\pm\left(\sum_{k=0}^n\frac{P_k}{k^2}\right)^{-\frac{1}{2}}\sum_{k=1}^n \frac{\sqrt{P_k}}{k}\ket{h_k}\right)+O(2^{-n}).
    \label{ctqw_true_eigenstates}
\end{equation}
Note, the second therm has high overlap with $\ket{\mathbf{0}}$ 
\begin{equation}
   \left(\sum_{k=0}^n\frac{P_k}{k^2}\right)^{-\frac{1}{2}}\bra{\mathbf{0}} \sum_{k=1}^n \frac{\sqrt{P_k}}{k}\ket{h_k} = \left(\sum_{k=0}^n\frac{P_k}{k^2}\right)^{-\frac{1}{2}} \sum_{k=1}^n \frac{P_k}{k}
\end{equation}
\begin{equation}
    =\left(\frac{4}{n^2}+O\left(n^{-3}\right)\right)^{-\frac{1}{2}}\left(\frac{2}{n}+O\left(n^{-2}\right)\right)
\end{equation}
\begin{equation}
    =1+O\left(n^{-1}\right),
\end{equation}
which justifies use of expressions \eqref{ctqw_eigenstates}. Notice, however, that while these expressions are only polynomially correct, this imprecision comes from approximating the second part of 
\eqref{ctqw_true_eigenstates} with state $\ket{\bm 0}$. The original form of the eigenstates \eqref{ctqw_true_eigenstates} remains exponentially precise. 

\section{QAOA sequence from the trotterized CTQW}
\label{appendix:CTQW_to_QAOA}
The second order Trotter sequence for the evolution of a Hamiltonian with two terms $H = H_1 + H_2$ is
\begin{align}
   S_2(t) &=  e^{H_2 \frac{t}{2}}e^{H_1 \frac{t}{2}}e^{H_1 \frac{t}{2}}e^{H_2 \frac{t}{2}}\\
   &= e^{H_2 \frac{t}{2}}e^{H_1 t}e^{H_2 \frac{t}{2}}.
\end{align}
Similarly, after grouping neighboring terms, higher order Suzuki sequences take the form 
\begin{equation}
    S_{q}(t)=\left(\prod_{v=1}^{5^{q/2-1}} e^{H_2ta(v)}e^{H_1tb(v)}\right)e^{H_2tc}, \label{trotter_seq_H1H2}
\end{equation}
where $a,b,c \in \mathbb{R}$. In the case of unstructured search, by setting $H_1=-iH_0$, $H_2=-iH_x \alpha^*$, a step $S_q(t/r)$ takes the form 
% we can safely ignore the first term in \eqref{trotter_seq_H1H2} as it only contributes a global phase when applied to the initial state $\ket{+}^{\otimes n}$,
\begin{equation}
    S_{q}(t/r)\ket{+}^{\otimes n}=e^{-i\frac{\alpha^*ntc}{r}}\prod_{v=1}^{5^{q/2-1}} e^{-iH_x\frac{\alpha^* t a(v)}{r}}e^{-iH_0\frac{tb(v)}{r}}\ket{+}^{\otimes n},
    \label{trotter_applied}
\end{equation}
which has $\Upsilon_q = 5^{q/2-1}$ stages. Using $r$ steps in the sequence \eqref{eq:formula_form}, and grouping terms in neighboring steps one gets

\begin{align}
   S_{q}^r(t/r)\ket{+}^{\otimes n} & = \left\{\left(\prod_{v=1}^{5^{q/2-1}} e^{-iH_x\frac{\alpha^*ta(v)}{r}}e^{-iH_0\frac{tb(v)}{r}}\right)e^{-iH_x{\frac{\alpha^*tc}{r}}}\right\}^r\ket{+}^{\otimes n}\notag\\
   &=e^{-i\alpha^*ntc'}\prod_{v=1}^{r5^{q/2-1}} e^{-iH_x\alpha^*ta'(v)}e^{-iH_0tb'(v)}\ket{+}^{\otimes n},
\end{align}

with $a', b', c' \in \mathbb{R}$, which recovers a QAOA sequence of depth $p = r\Upsilon_q = r5^{q/2-1}$.

{
\section{Numerical details}
\label{appendix:numerical}
The numerics presented in this paper were performed by simulating the QAOA circuits in the $n+1$ dimensional symmetric subspace. Due to the QAOA angles repeating across Trotter steps, we simulate these circuit by calculating powers of a step $S_q(t/r)$, which can be performed efficiently. 

In order to numerically approximate the optimal depth for a given $n$, $\epsilon$ and $q$ we:
\begin{enumerate}
\item Define the number of steps as $r_{dl}=d\cdot 2^{n/2-l}$, where $d,l$ are integers initially set to $d=1$, $l=0$. 
\item Numerically calculate $S_q^{r_{dl}}(t/r_{dl})$ and iteratively increase the value of $d$ one by one until finding $d=d'$ such that:   \begin{equation}
    |\bra{0}^{\otimes n}S_q^{r_{d'l}}(t/r_{d'l})\ket{+}^{\otimes n}|^2 \geq |\bra{0}^{\otimes n}U(\alpha^*,t^*)\ket{+}^{\otimes n}|^2 -\epsilon. \label{eq:acceptance_condition}
\end{equation}
\item Once condition \eqref{eq:acceptance_condition} is fulfilled we perform binary search. This is done by setting the new initial value of $d \rightarrow 2d'-1$ and $l\rightarrow l+1$, and  repeating step 2.
\item Steps 2 and 3 get repeated for a fixed number of iterations. Specifically for the numerics presented in the manuscript, we used 15 iterations.
\item The resulting approximated optimal depth is given by $p_\mathrm{numerical} = r_{dl}\cdot 5^{q/2-1}$.
\end{enumerate}
The code used for this paper is written in Python and is available on reasonable request.
}
\end{document}